\documentclass[sigconf, nonacm=true]{acmart}

\usepackage{subfig}
\usepackage{algorithmic}
\usepackage{algorithm}

\usepackage{amsmath}
\usepackage{mathtools}
\mathtoolsset{showonlyrefs=true}
\usepackage{amsthm}
\theoremstyle{definition}
\newtheorem{theorem}{THEOREM}
\newtheorem*{theorem*}{THEOREM}
\newtheorem{definition}[theorem]{DEFINITION}
\newtheorem{definition*}{DEFINITION}
\newtheorem{lemma}[theorem]{LEMMA}

\newcommand{\defeq}{\stackrel{\mathrm{def}}{=}}
\newcommand{\1}{\mbox{1}\hspace{-0.25em}\mbox{l}}

\AtBeginDocument{%
  }

\setcopyright{acmlicensed}
\copyrightyear{2018}
\acmYear{2018}
\acmDOI{XXXXXXX.XXXXXXX}

\acmConference[Conference acronym 'XX]{Make sure to enter the correct
  conference title from your rights confirmation emai}{June 03--05,
  2018}{Woodstock, NY}
\acmISBN{978-1-4503-XXXX-X/18/06}

\begin{document}

\title{Estimation of Graph Features Based on Random Walks \\ Using Neighbors' Properties}

\author{Tsuyoshi Hasegawa}
\affiliation{%
  \institution{Kyoto University}
  \country{Japan}
}

\author{Shiori Hironaka}
\affiliation{%
      \institution{Kyoto University}
  \country{Japan}
}

\author{Kazuyuki Shudo}
\affiliation{%
  \institution{Kyoto University}
  \country{Japan}
}
\renewcommand{\shortauthors}{Trovato et al.}

\begin{abstract}
  Using random walks for sampling has proven advantageous in assessing the characteristics of large and unknown social networks. Several algorithms based on random walks have been introduced in recent years. In the practical application of social network sampling, there is a recurrent reliance on an application programming interface (API) for obtaining adjacent nodes. However, owing to constraints related to query frequency and associated API expenses, it is preferable to minimize API calls during the feature estimation process. In this study, considering the acquisition of neighboring nodes as a cost factor, we introduce a feature estimation algorithm that outperforms existing algorithms in terms of accuracy. Through experiments that simulate sampling on known graphs, we demonstrate the superior accuracy of our proposed algorithm when compared to existing alternatives.
\end{abstract}




\maketitle

\section{Introduction}
Examining the graph structure of nodes and edges in online social networks (OSNs) is a significant challenge, prompting active research efforts to address this issue~\cite{ahn2007analysis,mislove2007measurement,gjoka2011practical,kwak2010twitter}. However, data access in conventional OSNs, like X\footnote{\url{https://twitter.com}}, is restricted~\cite{twitterapi, facebookapi, mastodonapi}, redering it nearly impossible to acquire and analyze the complete graph. Therefore, a pragmatic strategy involves estimating the graph's features by sampling a representative portion of the OSNs.

To estimate OSNs' features through sampling, leveraging random walks proves advantageous. Several random walk algorithms have been introduced for unbiased feature estimation~\cite{hardiman2013estimating,li2015on,ribeiro2012on,chen2016general,iwasaki2018estimating,matsumura2018average,nakajima2018estimating,matsumura2019metropolis}. Many OSNs offer application programming interfaces (APIs) that provide access to information about a user's follower or followee lists, specifically details about adjacent nodes~\cite{twitterapi,facebookapi,mastodonapi}. By iteratively selecting a node at random from the adjacent nodes obtained through the API and transitioning, random walk sampling on OSNs becomes feasible. Exploiting the inherent Markov property of random walks enables the computation of suitable weights for the obtained sample sequence, enabling the derivation of unbiased estimates for OSNs~\cite{li2015on,gjoka2011practical,lee2012beyond}. Uniform independent sampling based on node IDs is generally challenging owing to the unknown distribution of node IDs~\cite{chiericetti2016sampling}. Additionally, traversal methods like breadth-first sampling~\cite{kurant2011towards} cannot provide unbiased features owing to unknown biases in the acquired sample sequence.

APIs from common OSNs restrict the number of queries allowed per unit of time. Moreover, certain OSNs, such as X, have introduced charges for API usage. Therefore, estimating OSN features with minimal API calls is crucial, considering time and cost factors. Iwasaki et al.~\cite{iwasaki2018comparing} treated the API call count as a cost and compared it with existing random walk-based feature estimation algorithms. In our approach, we focus that many APIs allow obtaining both the list of adjacent nodes and the degree of those nodes simultaneously. We leverage this information to develop a more efficient algorithm. In the proposed algorithm, not confined to degree estimation, we can also estimate unbiased OSN features for any features obtained simultaneously when acquiring the list of adjacent nodes.  

In this study, we propose an algorithm for estimating features in OSNs using random walks and properties of adjacent nodes. Through simulation experiments, we demonstrate that our proposed algorithm attains the highest accuracy in estimating OSN features compared to existing methods. Our proposed method leverages the properties of adjacent nodes, which are obtained along with the adjacent node acquisition API, for unbiased feature estimation.

\section{Preliminaries}
In this section, we explain the foundational knowledge required for this study. In addition to several notions and definitions, the assumptions regarding APIs and OSN, as well as the fundamental concepts of Markov chains and random walks, also be discussed.
\subsection{Definitions and Notations}
In this study, we use the notation of a directed graph $G=(V,E)$ to represent the social graph. $V=\{v_1,v_2,...,v_n\}$ represents the set of nodes (users), with $n$ being the total number of nodes in the graph $(n=|V|)$. $E$ is the set of directed edges, depicting the following relationships. For every edge $(v_i, v_j)$, we introduce a set of edges and reverse edges by adding $(v_j, v_i)$, denoted as $E'$. When a directed edge $(v_i, v_j)$ exists, we refer to node $v_j$ as the friend of node $v_i$, and node $v_i$ as the follower of node $v_j$. For a node $v_i \in V$, we define the set of friends as $N_\mathrm{out}(v_i)=\{v_j\in V:(v_i,v_j)\}$ and the set of followers as $N_\mathrm{in}(v_i)=\{v_j\in V:(v_j, v_i)\}$. Additionally, $N(v_i)=N_\mathrm{out}(v_i)\cup N_\mathrm{in}(v_i)$. We also define the out-degree and in-degree of each node as $d_\mathrm{out}(v_i)=|N_\mathrm{out}(v_i)|$ and $d_\mathrm{in}(v_i)=|N_\mathrm{in}(v_i)|$, respectively. Moreover, we introduce the total degree as $d_\mathrm{sum}(v_i)=d_\mathrm{in}(v_i)+d_\mathrm{out}(v_i)$, and the mutual connections between followers and friends as $d_\mathrm{in\mathchar`-out}(v_i)=|N_\mathrm{out}(v_i) \cap N_\mathrm{in}(v_i)|$. 

We define the property of node $v_i$ as $a(v_i)$. Examples of the property $a(v_i)$ include the degree of $v_i$, the number of posts, and binary labels such as bot labels.

\subsection{Model}\label{subsec-model}

In this study, we focus on the APIs which enable acquiring the degree information (number of friends, number of followers) and properties to be estimated for each adjacent node when querying information about them. To clarify, when querying the list $N(v_i)$ of adjacent nodes for node $v_i$, we assume that the out-degree $d_\mathrm{out}(v_j)$, in-degree $d_\mathrm{in}(v_j)$ and property $a(v_j)$ of any node $v_j$ within $N(v_i)$ can be obtained simultaneously. In real OSNs, X and Mastodon offer APIs that adhere to this model~\cite{twitterapi, mastodonapi}.

We treats the frequency of acquiring adjacent nodes as a cost. We assume that a single instance of adjacent node acquisition allows for the simultaneous retrieval of $N_\mathrm{in}(v_i)$ and $N_\mathrm{out}(v_i)$. Regardless of the number of adjacent nodes, we assume that all adjacent nodes can be obtained at a fixed cost of 1.

We assume that the graph $G$ is weakly connected and remains static during the random walk. Additionally, upon transitioning to a node, we store information like its degree and properties. This includes maintaining a list of adjacent nodes along with their respective degrees and properties.

\subsection{Markov Chain Basics}
This section presents an overview of Markov chains. When estimating unbiased features in OSNs through sample sequences from random walks, it is imperative to correct biases by considering the steady-state distribution of the Markov chain. Let $\mathbf{P} = \{P_{i,j}\}_{i,j \in S}$ represent the transition probability matrix of a Markov chain in the state space $S$. The theorem below is applicable to the steady-state distribution of $\mathbf{P}$.

\begin{theorem}
\label{the:mar-sta}
    In the context of a distribution $\boldsymbol{\pi}=(\pi_i)_{i\in S}$, if the condition $\pi_j = \sum_{i\in S}\pi_i P_{i,j}$ is satisfied, it indicates that the distribution $\boldsymbol{\pi}$ serves as the steady-state distribution for a Markov chain governed by the probability transition matrix $\mathbf{P}$.
\end{theorem}
Subsequently, Theorem \ref{the:mar-elg}~\cite{jones2004markov, levin2017markov} is applicable to the steady-state distribution $\boldsymbol{\pi}$.
\begin{theorem}
\label{the:mar-elg}
     In the case of an ergodic Markov chain determined by $\mathbf{P}$, there is a singular and unique steady-state distribution $\boldsymbol{\pi}$.
\end{theorem}
To illustrate the convergence of sample sequences from a Markov chain to the desired features for estimation, we refer to Theorem \ref{the:mar-str}~\cite{chen2016general, lee2012beyond}, which relies on the strong law of large numbers.
\begin{theorem}
\label{the:mar-str}
     Consider a finite state space $S$ with an ergodic Markov chain $\{X_s\}$, and let $\boldsymbol{\pi}$ be its steady-state distribution. For any function $f:S\rightarrow\mathbb{R}$, the following holds as time $t$ approaches infinity, irrespective of the initial state.
    \[
    \frac{1}{t}\sum_{s=1}^{t}f(X_s) \rightarrow \sum_{i\in S}\pi _i f(i) \text{ almost surely (a.s.)}
    \]
\end{theorem}

\subsection{Random Walk Sampling}
This study presents a novel sampling technique based on random walks, elucidating the methodology for estimating unbiased features in OSNs. This section provides an overview of a simple random walk (SRW).

SRW sampling entails the progression from an initially selected node to a randomly chosen neighboring node. Representing the transition probability from node $v_i$ to node $v_j$ in SRW as $P_{i, j}$,
\begin{equation}
    P_{i, j} = 
    \begin{cases}
    \frac{1}{d_\mathrm{sum}(v_i)} & v_j \in N(v_i),\\
    0 & \text{otherwise.}
    \end{cases}
\end{equation}

In the context of random walk sampling, the transition probabilities for each node are determined through mathematical analysis. Specifically, the distribution $\boldsymbol{\pi}$ illustrating transition probabilities to each node after $t$ steps in SRW is represented as $\boldsymbol{\pi} = (Pr[x_t=1], Pr[x_t=2],...,Pr[x_t=n])$, where $Pr[A]$ denotes the probability of event $A$. If we denote the $i$-th element of $\boldsymbol{\pi}$ as $\pi_i$, it is established that $\pi_i$ converges to $d_\mathrm{sum}(v_i)/2|E|$~\cite{aldous2002reversible}. As a result, in SRW, the transition probability to each node is directly proportional to its degree. Exploiting this characteristics of SRW allows us to obtain unbiased features for the entire graph~\cite{gjoka2011practical}.

Many studies on sampling methods involving random walks have traditionally assessed accuracy based on the size of the sample sequence~\cite{gjoka2011practical, lee2012beyond, hardiman2013estimating}. In contrast, our study adopts a different perspective by considering the acquisition of a list of neighboring nodes as a cost. This approach, aligned with the methodology of Iwasaki et al.~\cite{iwasaki2018comparing}, stems from the realization that obtaining the list of neighboring nodes can pose a practical bottleneck in OSN sampling. As outlined in Section \ref{subsec-model}, the acquisition of neighboring nodes is assumed to carry a cost of 1, and the out-degree, in-degree and properties of each node in the obtained list of neighboring nodes can be determined. Of note, this information is not utilized in the SRW process.

\section{Proposed Method}
In this section, we present a sampling algorithm based on random walks that utilizes the properties of each acquired adjacent node during the process of obtaining adjacent nodes. We discuss the Markov chain aspect within our proposed method to elucidate the algorithms employed for estimating these features.

The features on the OSN, which our proposed method can estimate are derived from properties obtained concurrently with the acquisition of adjacent nodes. In specific terms, when retrieving the list of adjacent nodes for node $v_i$, if we can get the property $a(v_j)$ for any node $v_j$ in the adjacent node list $N(v_i)$, our proposed method can estimate average and distribution of the property.

\subsection{Probabilistic Addition of Adjacent Nodes to the Sample Sequence}
\begin{figure}[t]
    \centering
    \includegraphics[width = 0.9\linewidth]{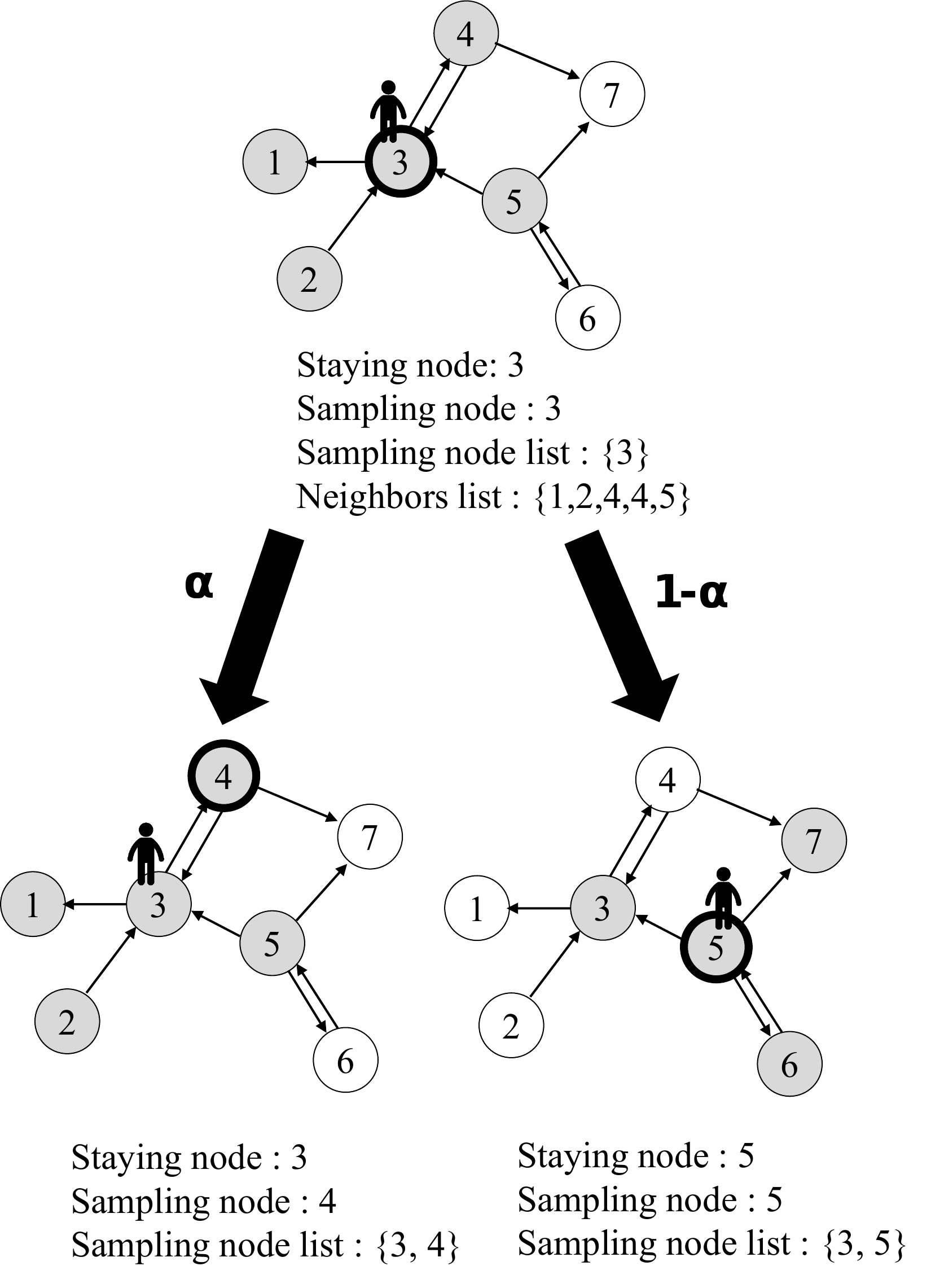}
    \caption{Overview of proposed method, while gray nodes denote nodes capable of acquiring degree information and properties.}
    \label{fig:ours}
\end{figure}

\begin{figure}[!t]
    \begin{algorithm}[H]
        \caption{Proposed sampling algorithm}
        \label{alg:ours}
        \begin{algorithmic}[1]    
        \REQUIRE $n_0$ : initial node, $\alpha$ : parameter$(0\leq \alpha < 1)$, $b$ : number of queries
        \ENSURE $sampling\_node\_list$ : sample sequence
        \STATE $v_i \leftarrow n_0$
        \STATE $sampling\_node\_list \leftarrow \{v_i\} $
        \STATE $staying\_node\_list \leftarrow \{ {v_i}\}$
        \STATE $query\_count \leftarrow 1$
        \WHILE{$query\_count < b$}
            \STATE $N(v_i) \leftarrow $ A list obtained by concatenating $N_\mathrm{out}(v_i)$ and $N_\mathrm{in}(v_i)$
            \STATE $p\leftarrow$ Random number generated between 0 and 1
            \WHILE{$p < \alpha$}
                \STATE $v_j \leftarrow$ A node uniformly selected at random from $neighbors$
                \STATE Append $v_j$ to $sampling\_node\_list$
                \STATE $p\leftarrow$ Random number generated between 0 and 1
            \ENDWHILE
            \STATE $v_i \leftarrow$ A node uniformly selected at random from $neighbors$
            \IF{$v_i \notin staying\_node\_list$}
            \STATE $query\_count \leftarrow query\_count + 1$
            \ENDIF
            \STATE Append $v_i$ to $sampling\_node\_list$
            \STATE Append $v_i$ to $staying\_node\_list$
            
        \ENDWHILE
        \RETURN $sampling\_node\_list$
        \end{algorithmic}
    \end{algorithm}
\end{figure}

We present an overview of our proposed sampling algorithm in Figure \ref{fig:ours}, while the detailed steps are outlined in Algorithm \ref{alg:ours}. The fundamental transition method in our algorithm closely resembles a standard random walk. However, after transitioning to node $v_i$ and its addition to the sample sequence, a distinctive element is introduced. With a probability $\alpha$, instead of the typical transition, we incorporate a randomly chosen node from the acquired adjacent nodes into the sample sequence. Here, $\alpha$ is a parameter in the range $0 \leq \alpha < 1$, and the transition from node $v_i$ to node $v_j$ occurs with a probability of $1-\alpha$. In this paper, we refer to the node reached after the transition as the \textit{staying node}, the node added to the sample sequence as the \textit{sampling node}, the operation executed with a probability $\alpha$ as \textit{adjacent node sampling}, and the operation performed with a probability $1-\alpha$ as \textit{transition sampling}. In adjacent node sampling, the staying node remains unchanged, and the sampling node is randomly selected from the adjacent nodes of the staying node. In transition sampling, the staying node is updated, and the sampling node becomes identical to the staying node.

In this study, we treat OSNs as directed graphs, allowing transitions and sampling to occur even on reverse edges. This occurs because, in real OSNs when transitioning between users, it is possible to choose users who transition from both followers and friends. Consequently, the list of adjacent nodes for a node $v_i$ is a combination of its in-neighbors $N_\mathrm{in}(v_i)$ and out-neighbors $N_\mathrm{out}(v_i)$. In this context, the list of adjacent nodes for node $v_i$ is an array that allows duplicates. Specifically, for a node $v_j$, if $v_j \in N_\mathrm{out}(v_i) \cap N_\mathrm{in}(v_i)$, then node $v_i$ has two occurrences of node $v_j$ in its list of adjacent nodes. Because the node to transition or sample is uniformly and randomly chosen from the list of adjacent nodes, a node appearing twice in the list has double the probability of being chosen compared to a node that appears only once.

In our study, sampling adjacent nodes is a cost-free process because we define the acquisition frequency of adjacent nodes as the cost. This operation randomly adds nodes to the sample sequence from the existing list of adjacent nodes, incurring no new expenses. In contrast, transition sampling involves obtaining the adjacent nodes of the transition destination node, incurring a cost. However, we assume that the information about once-acquired adjacent nodes is stored (Section \ref{subsec-model}). Consequently, if a node previously sampled through transition sampling is sampled again, no additional cost is incurred. Therefore, the frequency of acquiring adjacent nodes corresponds to the count of new transition sampling events.

We use the term \textit{query count} to represent the number of times the list of adjacent nodes is obtained through the new transition sampling. The proposed sampling algorithm stops when the number of acquired adjacent nodes reaches a specified query limit $b$. This limit can be established based on the OSN's API specifications, ensuring a reasonable number of queries within a given time frame. Notably, a higher query limit $b$ generally correlates with improved accuracy in estimation accuracy. The relationship between the query limit and estimation accuracy are elucidated in Section \ref{sec-exe}.

\subsection{Markov Chains in the Proposed Sampling Algorithm}
To discuss the Markov chain in the proposed sampling algorithm, we define the state space as follows.
\begin{definition}
\label{def:omg}
    \begin{equation}
        \Omega \defeq  \{(v_i,v_j) : v_i,v_j \in V \text{ s.t.}((v_i,v_j)\in E') \text{ or } (v_i = v_j)\}
    \end{equation}
\end{definition}

Consider $X_t$ as the node at step $t$, and $X_t'$ as the sampling node. Define $Z_t \defeq (X_t,X_t')\in \Omega$. The proposed sampling algorithm can be seen as a Markov chain $\{Z_t\in\Omega: t=1,2,...\}$ on the state space $\Omega$. For simplicity, let us use $e_{ij}$ to represent the state $(v_i,v_j)\in \Omega$. In this context, $e_{ij}$ denotes adjacent node sampling when $i\ne j$, indicating that the staying node is $v_i$, and the sampling node is $v_j$. The condition of $i = j$, denotes transition sampling, with both the stay and sampling nodes being $v_i (= v_j)$.

Consider $\mathrm{P}\defeq \{p(e_{ij},e_{lk}) : e_{ij},e_{lk} \in \Omega\}$ as the transition probability matrix of the Markov chain $\{Z_t \in \Omega : t = 1,2,...\}$ on the state space $\Omega$. The following theorem holds.

\begin{theorem}
\label{the:mal}
The stationary distribution $\boldsymbol{\pi}$ of $\mathrm{P}$ uniquely exists.
\end{theorem}

\begin{proof}
First, because $|V|<\infty$, it follows that $|\Omega|<\infty$. In our study, we assume that the target graph is a weakly connected directed graph. By considering transition in the reverse edge direction, we can transition to any node and sample any adjacent node. This allows reaching any state from any state within $\Omega$, eliminating periodicity. Thus, $\mathrm{P}$ is ergodic, quaranteeing a unique stationary distribution $\boldsymbol{\pi}$, as outlined in Theorem \ref{the:mar-elg}.
\end{proof}

Subsequently, a function $m$ is defined for the state $e_{ij}$ as follows.
\begin{definition}
    \begin{equation}
        m(e_{ij}) \defeq \1 _{v_j\in N_\mathrm{out}(v_i)} + \1 _{v_j\in N_\mathrm{in}(v_i)} 
    \end{equation}
\end{definition}

The function $m$ represents the number of directed edges between the staying node $v_i$ and the sampling node $v_j$ in the state $e_{ij}$. Thus, it holds that $m(e_{ij}) = m(e_{ji})$.

This implies that in our proposed sampling algorithm, while remaining at the node $v_i$, the adjacent node sampling probability $p_{neighbor}$ of staying at node $v_i$ and sampling adjacent node $v_j$ is presented as follows.

\begin{definition}
    \label{def:pn}
    \begin{equation}
    p_{neighbor}(v_j;v_i)=\alpha \cdot \frac{m(e_{ij})}{d_\mathrm{sum}(v_i)}
\end{equation}
\end{definition}

Likewise, when staying at node $v_i$, the transition sampling probability $p_{walk}$ of transitioning and sampling an adjacent node $v_j$ is determined as follows.

\begin{definition}
\label{def:pw}
    \begin{equation}
    p_{walk}(v_j;v_i)=(1-\alpha) \cdot \frac{m(e_{ij})}{d_\mathrm{sum}(v_i)}
\end{equation}
\end{definition}

The transition probability matrix $\mathbf{P}$ for the proposed sampling algorithm is given as follows.

\begin{definition}
\label{def:p}
\begin{equation}
    P(e_{ij},e_{lk})=
    \begin{cases}
    p_{neighbor}(v_k;v_i) & i=l,\\
    p_{walk}(v_k;v_i) & l=k, v_l\in N(v_i), \\
    0 & \text{otherwise.}
    \end{cases}
\end{equation}
\end{definition}

Now, we establish the following theorem regarding the stationary distribution of the Markov chain $\{Z_t : t = 1, 2, ...\}$ in the proposed sampling algorithm.

\begin{theorem}
\label{the:sta}
    When the transition probability matrix of the Markov chain $\{Z_t : t = 1, 2, ...\}$ within the state space $\Omega$ is defined as presented in Definition \ref{def:p}, the stationary distribution $\boldsymbol{\pi}$ of this Markov chain is given as follows.
    \begin{equation}
    \pi (e_{ij}) = 
    \begin{cases}
        \alpha \cdot \frac{m(e_{ij})}{2|E|} & i\neq j \\
        (1-\alpha) \cdot \frac{d_\mathrm{sum}(v_i)}{2|E|} & i=j 
    \end{cases}
\end{equation}
\end{theorem}
\begin{proof}
Refer to the appendix \ref{ape-sta}.
\end{proof}

\subsection{Feature Estimation}\label{subsec-est}
We introduce an algorithm for feature estimation in the context of the proposed sampling algorithm. In this section, we present the algorithm for estimating features using in the proposed sampling method. The primary objective is to apply weighting to the sample sequence generated by the proposed transition sampling algorithm, creating a process for estimating the expected value of features on the OSN.

Let $f$ be any function $f : V \rightarrow \mathbb{R}$, and consider the uniform distribution $\mathbf{u} \defeq [u(1),u(2),...,u(n)] = [1/n,1/n,...,1/n]$. The expected value $\mathbb{E}_u(f)$ of the feature to be estimated on the OSN is given as follows.

\begin{definition}
\label{def:eu}
    $\mathbb{E}_u(f) \defeq \frac{1}{n}\sum_{v\in V}f(v)$
\end{definition}

To obtain the expected value of the desired feature on the OSN, this can be accomplished by appropriately defining the function $f$. For instance, if you wish to estimate the out-degree distribution $\mathbb{P}\{D_G = d\}, (d = 1, 2, ..., n-1)$ of a graph $G$, selecting the function $f(v) = \1_{d_\mathrm{out}(v) = d}$ would be suitable.

To define a reweighting process for obtaining $\mathbb{E}_u(f)$ from the sample sequence $\{Z_s\}_{s=1}^{t}$ generated by the proposed sampling algorithm, a new function $g : \Omega \rightarrow \mathbb{R}$ is introduced for the function $f$ as follows.

\begin{definition}
\label{def:g}
    $g(e_{ij}) \stackrel{\mathrm{def}}{=}f(v_j)$
\end{definition}

The function $g$ applies the function $f$ to the sampling nodes in the sample $Z_t$.

Next, we define a weighting function as follows.

\begin{definition}
\label{def:w}
    $w(e_{ij}) \stackrel{\mathrm{def}}{=} \frac{1}{d_\mathrm{sum}(v_j)}$
\end{definition}

In this study, we assume that the number of friends $d_\mathrm{out}(v_j)$ and followers $d_\mathrm{in}(v_j)$ of sampling nodes are accessible. Consequently, we can obtain $d_\mathrm{sum}(v_j)$ for the sampling node $v_j$.

Here, the following theorem holds.

\begin{theorem}
\label{the:fea}
For the sample sequence $\{Z_s\}_{s=1}^{t}$ obtained from the proposed transition algorithm, as $t\rightarrow \infty$,
\begin{equation}
        \frac{\sum_{s=1}^tw(Z_s)g(Z_s)}{\sum_{s=1}^tw(Z_s)} \rightarrow \mathbb{E}_u(f) \text{ a.s.}
    \end{equation}
\end{theorem}

\begin{proof}
    Refer to the appendix \ref{ape-fea}.
\end{proof}

As stated in Theorem \ref{the:fea}, for the sample sequence $\{Z_s\}_{s=1}^t$ obtained by the proposed sampling algorithm, the estimator \linebreak $\sum_{s=1}^tw(Z_s)g(Z_s)/\sum_{s=1}^tw(Z_s)$ converges to the expected value of the desired feature on the OSN. Though the sample $Z_t$ contains information about both sampling node $X_t'$ and the staying node $X_t$, retaining information solely about sampling node $X_t'$ is sufficient for feature estimation. Therefore, in Algorithm \ref{alg:ours}, information about the staying node $X_t$ is not returned.

The proposed method can estimate the unbiased features regarding any properties that can be obtained when acquiring adjacent nodes. This is because the estimable feature $g(e_{ij})$ for the sample $e_{ij}$ added through adjacent node sampling is derived from the property $a(v_j)$ obtained at the same time during the acquisition of adjacent nodes. The specific content of this properties varies based on the OSN's API specifications. For instance, in X, which provides degree information~\cite{twitterapi}, one can estimate the average degree and degree distribution. Similarly, in Mastodon, which provides bot rates and post counts~\cite{mastodonapi}, these parameters can also be estimated.

Thus, we have successfully developed a weighting process to estimate the expected value of features on the OSN for the sample sequence acquired through the proposed sampling algorithm.

\section{Experiment}\label{sec-exe}
\begin{table}[t]
\centering     
\caption{Dataset Overview}
\begin{tabular}{|c|c|c|c|} 
\hline                
Network & Type & Nodes & Edges \\     
\hline                
ego-Twitter & Social Network & 81,306 & 1,768,149 \\
soc-Slashdot & Social Network & 82,168 & 948,464 \\
Amazon & Product Network & 262,111 & 1,234,877 \\
DBA model & Generation Network & 100,000 & 1,000,000 \\
\hline
\end{tabular}
\label{tab:abs}
\end{table}

\begin{figure*}[t]
  \subfloat[ego-Twitter]{\includegraphics[scale=0.29]{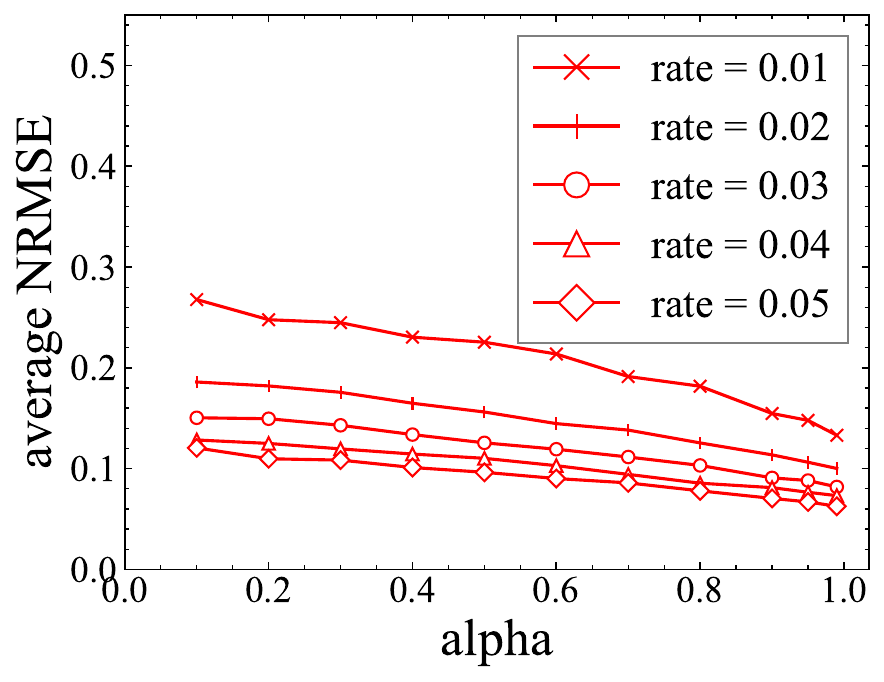}}%
  \subfloat[Slashdot]{\includegraphics[scale=0.29]{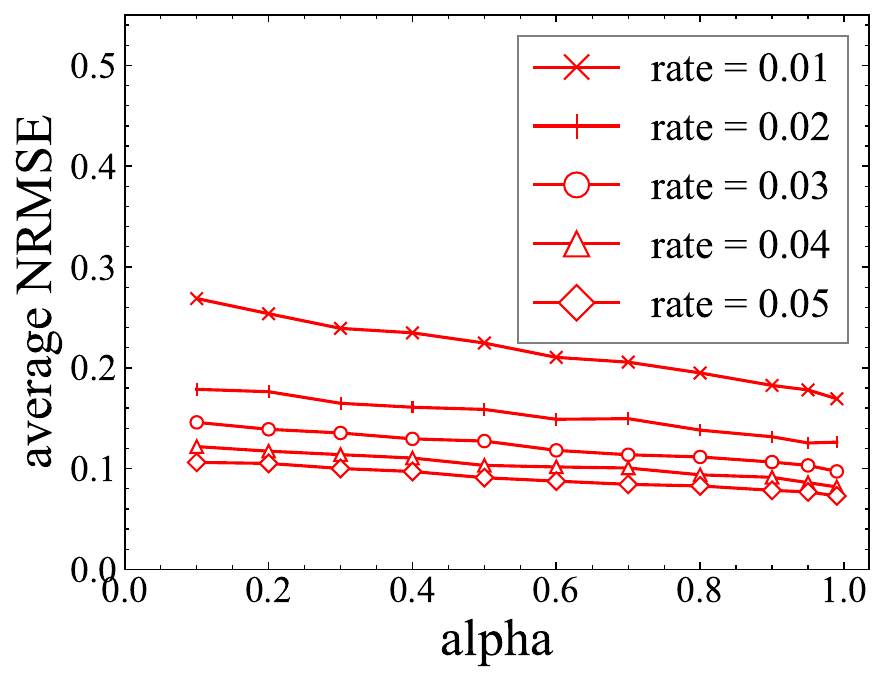}}%
  \subfloat[Amazon]{\includegraphics[scale=0.29]{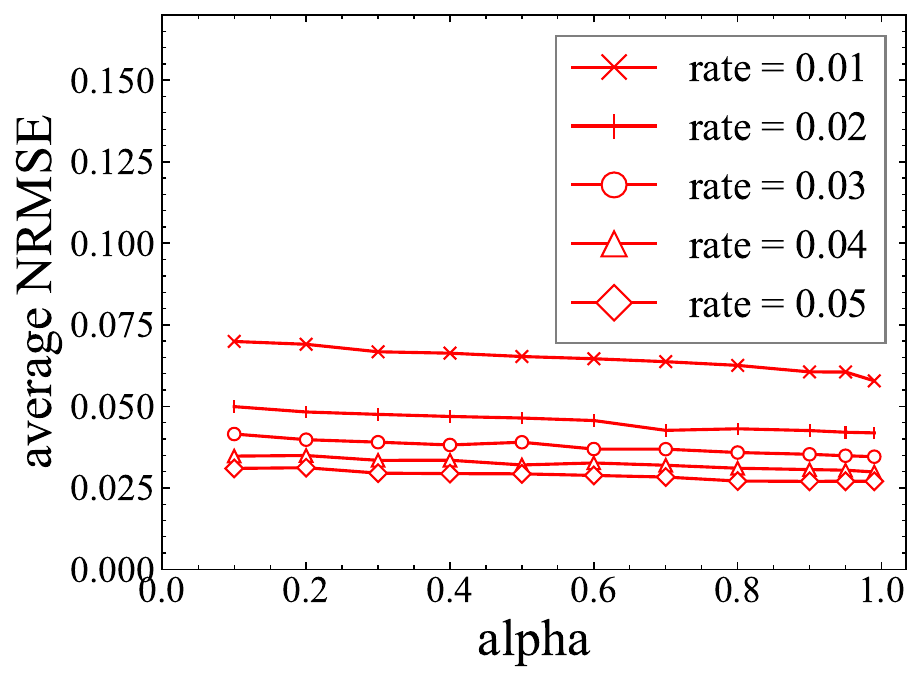}}%
  \subfloat[DBA model]{\includegraphics[scale=0.29]{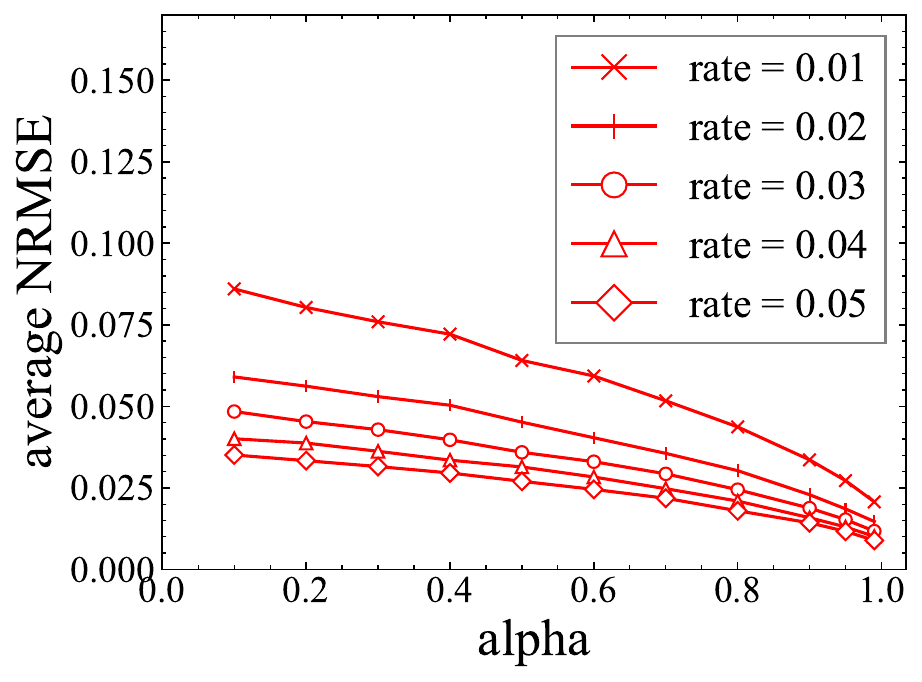}}%
\caption{Average NRMSE for each feature categorized by query rate at each $\alpha$.}
\label{fig:alpha} 
\end{figure*}

We assess the accuracy of the proposed method across various networks.  In real-world OSN sampling, the target graph is often unknown. However, for these experiments, we conduct sampling simulations on known graphs to facilitate evaluation.\footnote{Our simulation code is available at \url{https://github.com/XXXXXXXX}}

\subsection{Experimental Setup}\label{subsec-preexe}
\textbf{Dataset: }Our experiments utilized three datasets from the Stanford Large Network Dataset Collection~\cite{snap}, and we also employed the Directed-Barabasi-Albert model (DBA model)~\cite{posfai2016network}, a generative model for complex networks. The target graphs are directed, and we focus on the maximum weakly connected component. The DBA model extends the Barabasi-Albert model~\cite{posfai2016network} to directed graphs by introducing directed edges from one node to another with the following probability.
\begin{equation}
\prod d_\mathrm{in}(v_i)
 = \frac{d_\mathrm{in}(v_i) + A}{\Sigma _j(d_\mathrm{in}(v_j) + A)}.
\end{equation}
The parameter $A$ was set to $1$. Table \ref{tab:abs} provides an overview of each dataset.

\noindent\textbf{Simulation: }For the proposed sampling algorithm, the initial nodes are randomly selected from the graph, this process is independent for each simulation. All sampling simulations are independently conducted 1000 times. Query count $b$ is chosen by the proportion of the total number of nodes in the graph.

\noindent\textbf{Evaluation metrics: }The evaluation metric for each feature estimation is the Normalized Root Mean Square Error (NRMSE). NRMSE is widely employed in related studies to assess the accuracy of estimated values~\cite{lee2012beyond,hardiman2013estimating,iwasaki2018estimating,iwasaki2018comparing}; lower values indicate superior performance. It is computed as follows, where $x$ represents the true value of the feature, $\hat{x_i}$ is the estimated value of the feature in the $i$-th sampling simulation, and $N$ is the number of simulations.
\begin{equation*}
    \mathrm{NRMSE} = \frac{1}{x}\sqrt{\frac{1}{N}\sum_{i=1}^N(x-\hat{x_i})^2}.
\end{equation*}

\noindent\textbf{Features: }The features under investigation in our experiments include average out-degree, random label rate, high-degree label rate, and low-degree label rate. Random, high-degree, and low-degree labels are binary labels synthetically introduced to the dataset, simulating scenarios akin to bot labels in real OSNs. A random label is assigned randomly; a high-degree label is assigned with a probability of $\hat{d}(v_i)/\Sigma_{v_i\in V}\hat{d}(v_i)$ when $\hat{d}(v_i) = \text{max}(d_\mathrm{in}(v_i),d_\mathrm{out}(v_i))$, and a low-degree label is assigned with a probability of \linebreak $(1/\hat{d}(v_i))/\Sigma_{v_i \in V}(1/\hat{d}(v_i))$. Labels are repetitively assigned to selected nodes based on the corresponding probability until the labeled node proportion reaches $10\%$ of the total. The introduction of labeled nodes enables simulations to estimate the proportion of nodes with characteristics like randomly assigned labels, labels more prevalent in high-degree nodes, and labels more prevalent in low-degree nodes, reflecting real OSN scenarios. The labeling method follows the approach proposed by Fukuda et al.~\cite{fukuda2022estimating}.

\subsection{Relationship between $\alpha$ and Estimation Accuracy}
In the proposed method, we conducted experiments to investigate the relationship between the probability $\alpha$ of performing neighboring node sampling and estimation accuracy. Across each dataset, we varied $\alpha$ from $0.1$ to $0.9$ in increments of $0.1$. Additionally, we tested $\alpha$ values of $0.95$ and $0.99$. The query count was chosen $0.1\%$ to $0.5\%$ of the total number of nodes in the graph, adjusting in increments of $0.1\%$.

Figure \ref{fig:alpha} illustrates the average NRMSE for each feature discussed in Section \ref{subsec-preexe}. Results are presented for query counts representing $1\%, 2\%, 3\%, 4\%,$ and $5\%$ of all node counts. Figure \ref{fig:alpha} shows that as $\alpha$ approaches 1, the average NRMSE decreases across all graphs and query numbers. Therefore, the parameter $\alpha$ should be set to the largest possible value that is still less than 1.
\begin{figure*}[p]
      \subfloat[ego-Twitter]{\includegraphics[scale=0.29]{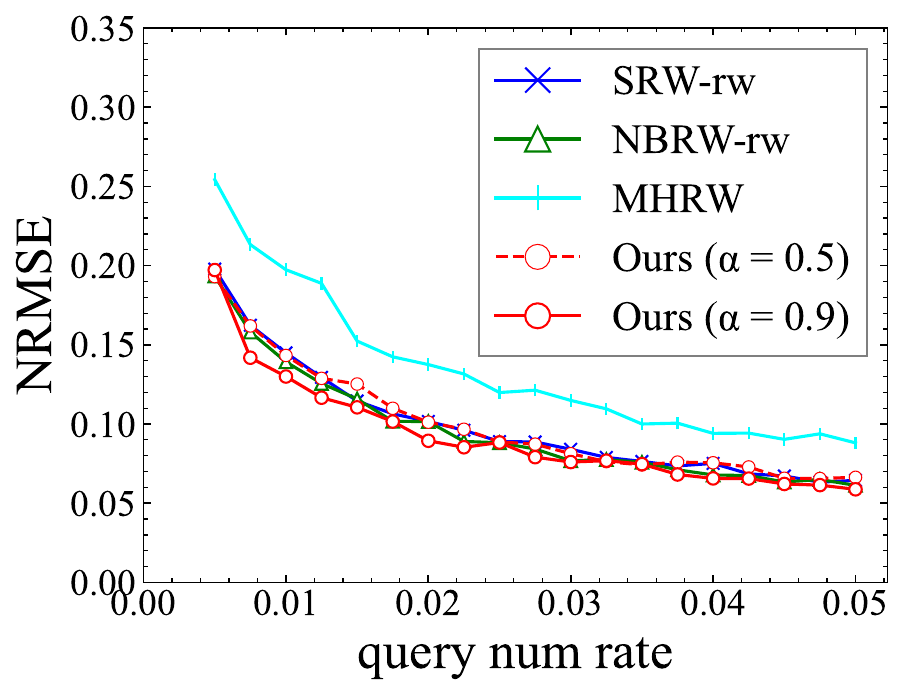}}
      \subfloat[Slashdot]{\includegraphics[scale=0.29]{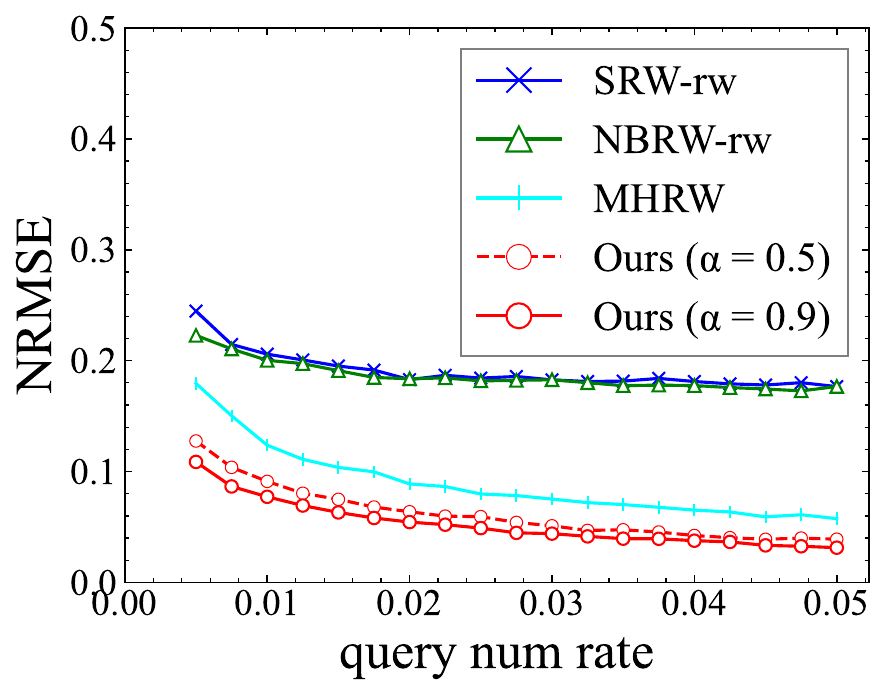}}
      \subfloat[Amazon]{ \includegraphics[scale=0.29]{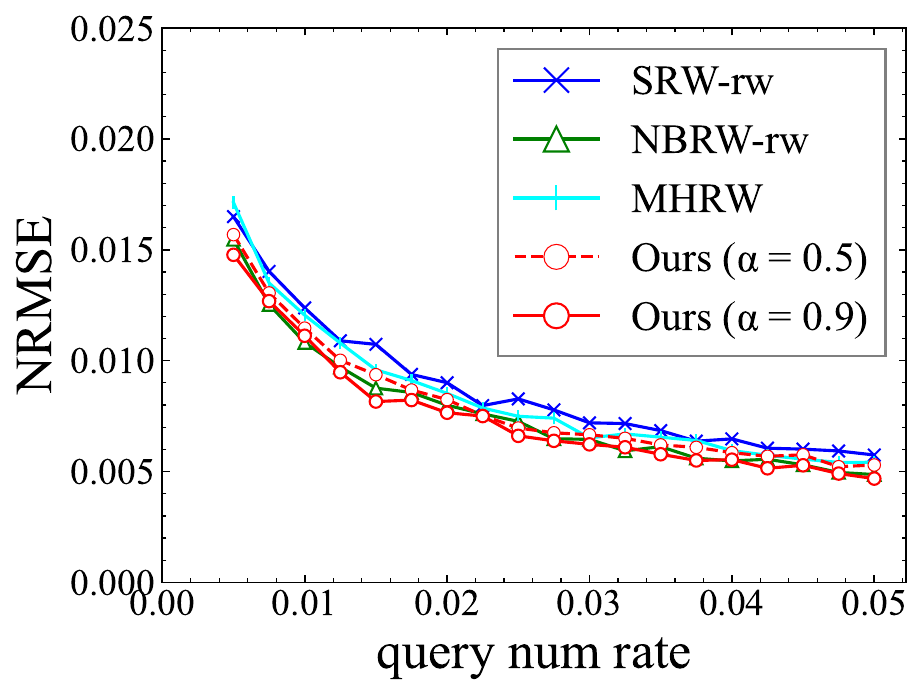}}
      \subfloat[DBA model]{\includegraphics[scale=0.29]{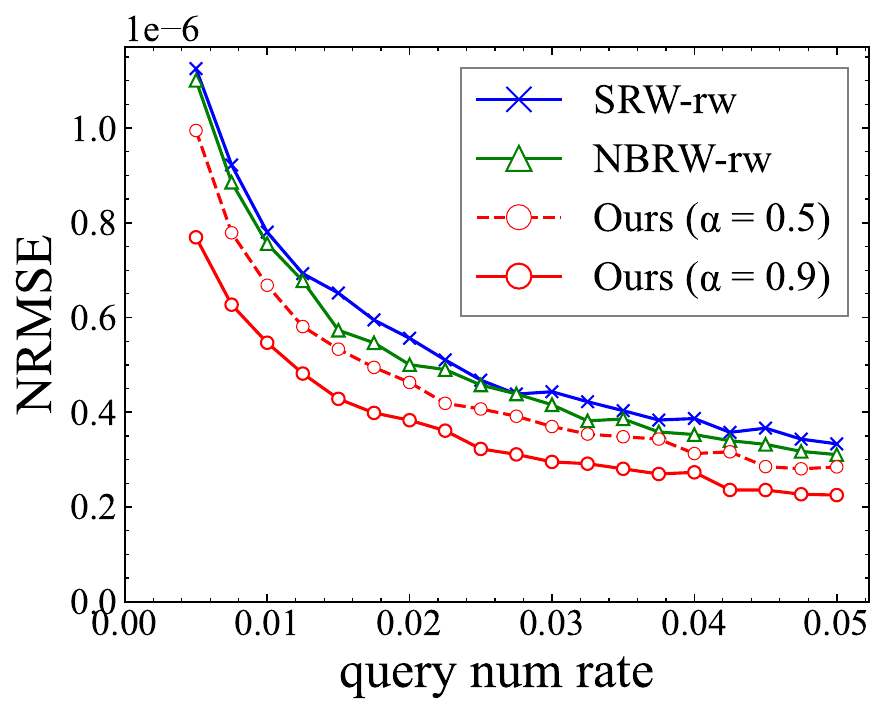}}
    \caption{NRMSE for out-degree estimation.}
    \label{fig:degree} 
  \end{figure*}

\begin{figure*}[p]
      \subfloat[ego-Twitter]{\includegraphics[scale = 0.29]{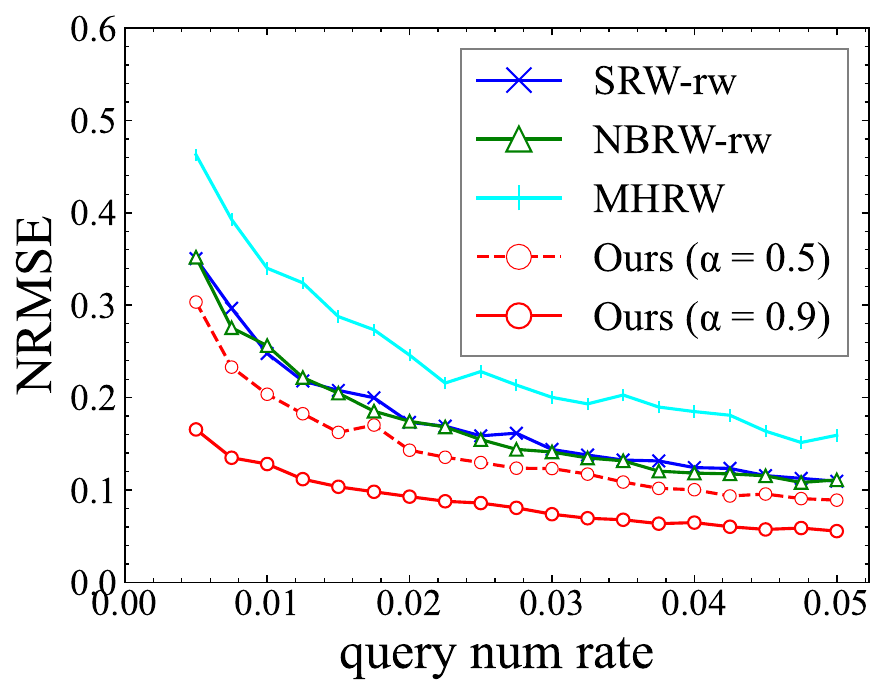}}
      \subfloat[Slashdot]{\includegraphics[scale=0.29]{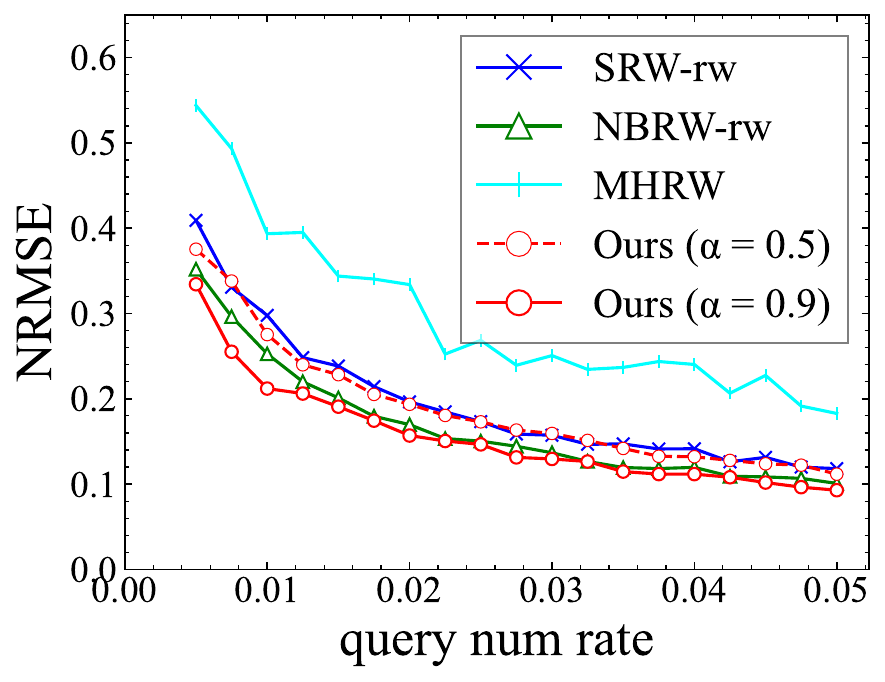}}
      \subfloat[Amazon]{\includegraphics[scale=0.29]{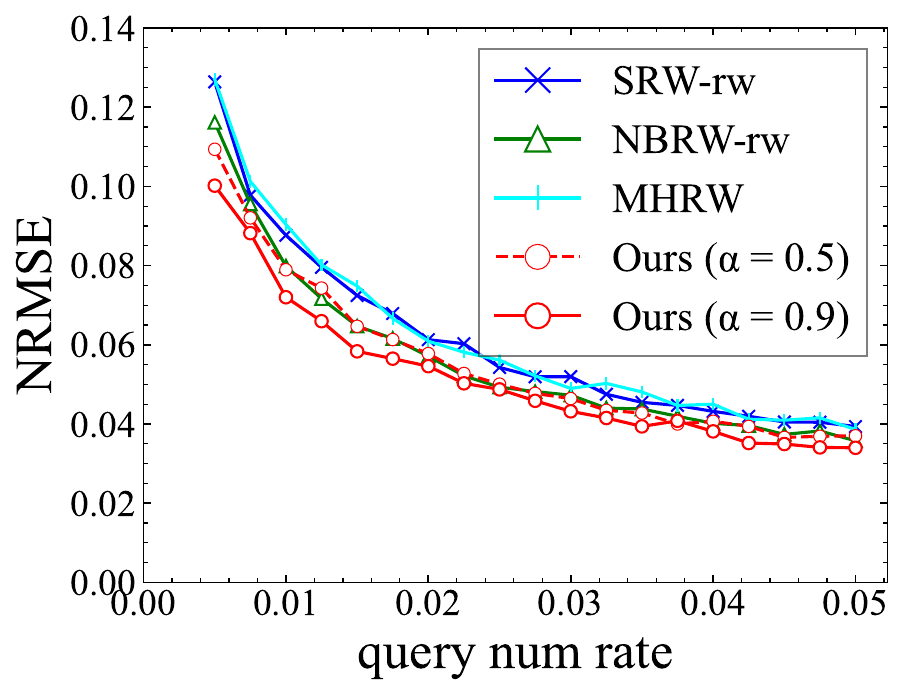}}
      \subfloat[DBA model]{ \includegraphics[scale=0.29]{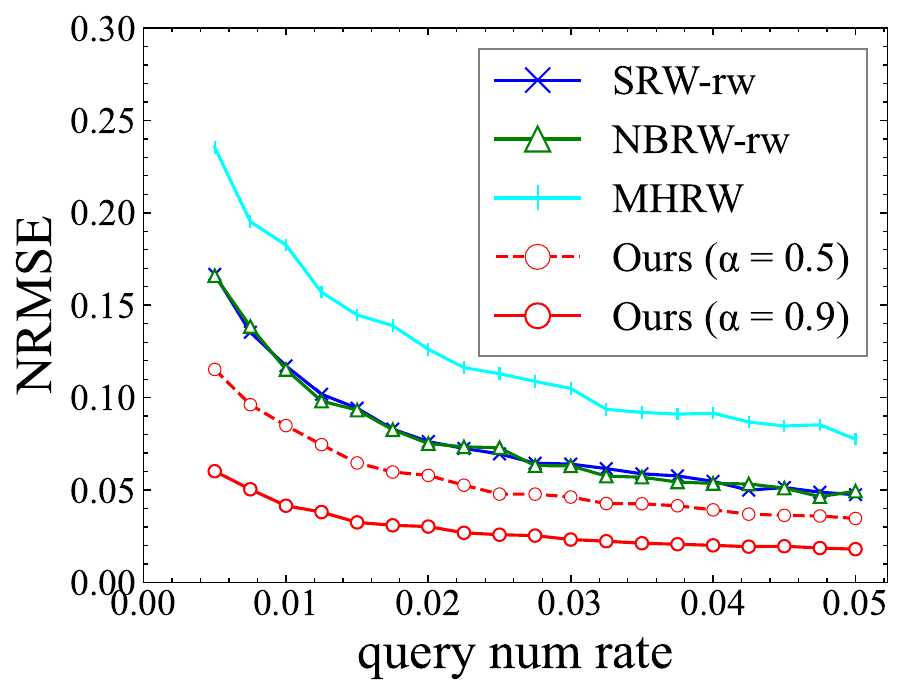}}
    \caption{NRMSE for random label estimation.}
    \label{fig:randlabel} 
  \end{figure*}

\begin{figure*}[p]
      \subfloat[ego-Twitter]{\includegraphics[scale=0.29]{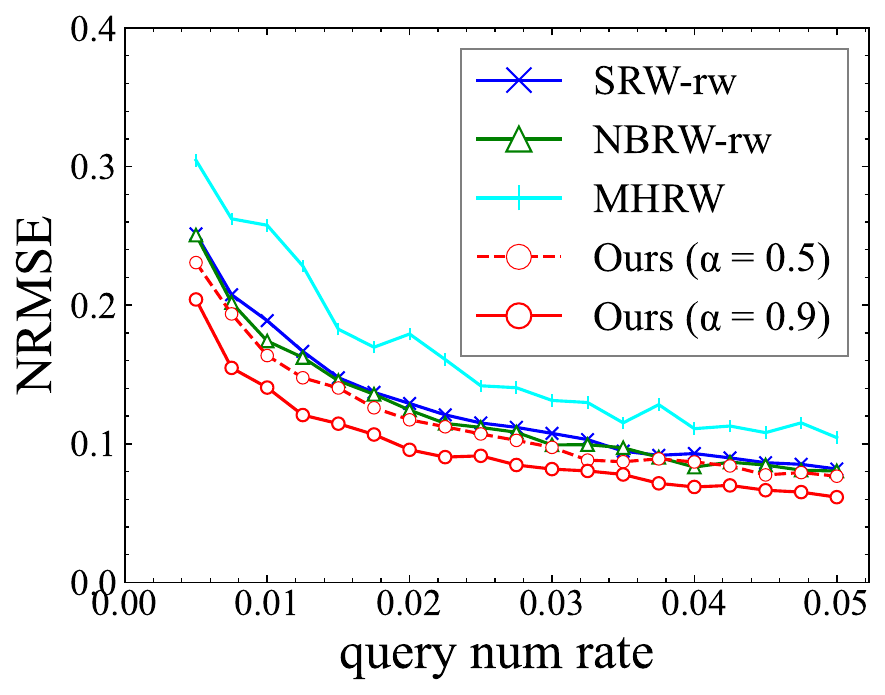}}
      \subfloat[Slashdot]{\includegraphics[scale=0.29]{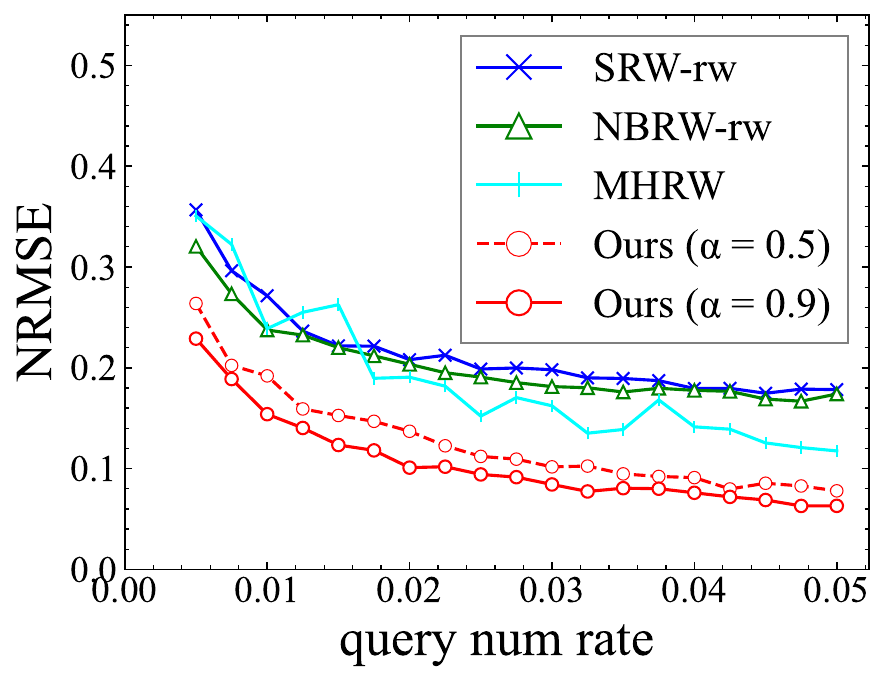}}
      \subfloat[Amazon]{\includegraphics[scale=0.29]{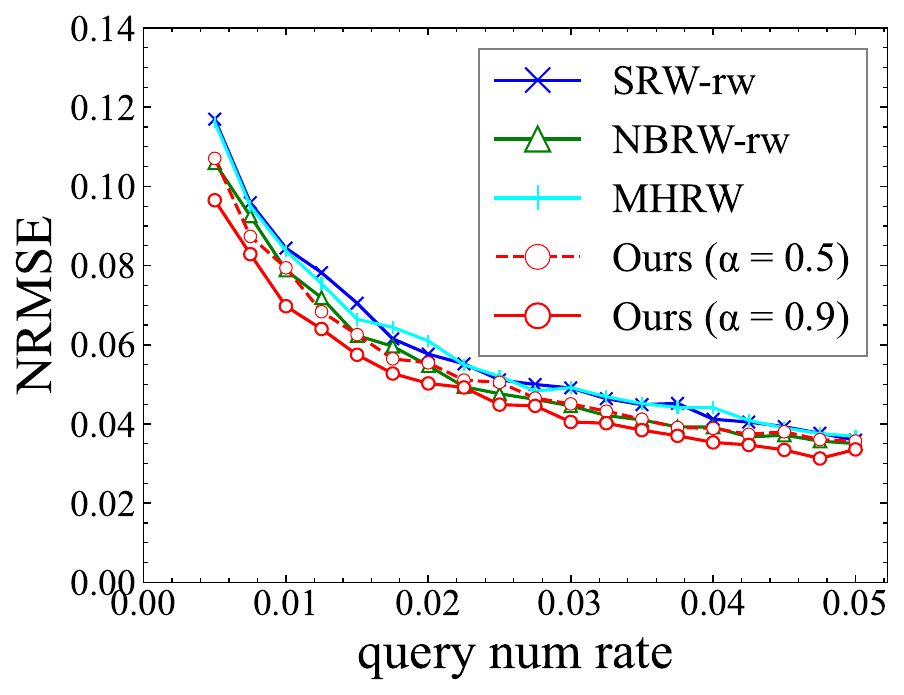}}
      \subfloat[DBA model]{ \includegraphics[scale=0.29]{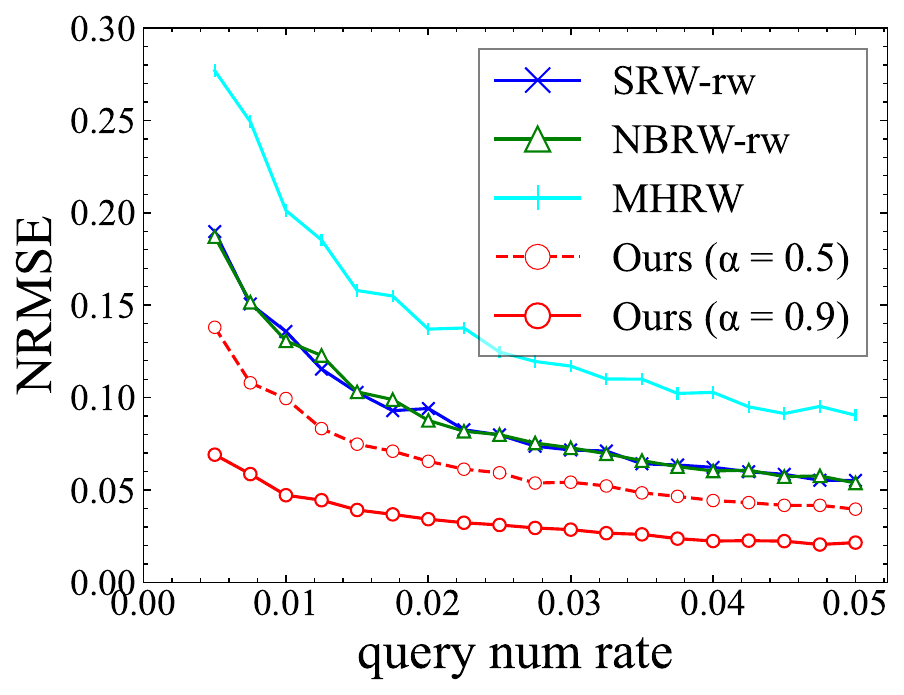}}
    \caption{NRMSE for high degree label estimation.}
    \label{fig:highlabel} 
  \end{figure*}

\begin{figure*}[p]
      \subfloat[ego-Twitter]{\includegraphics[scale=0.29]{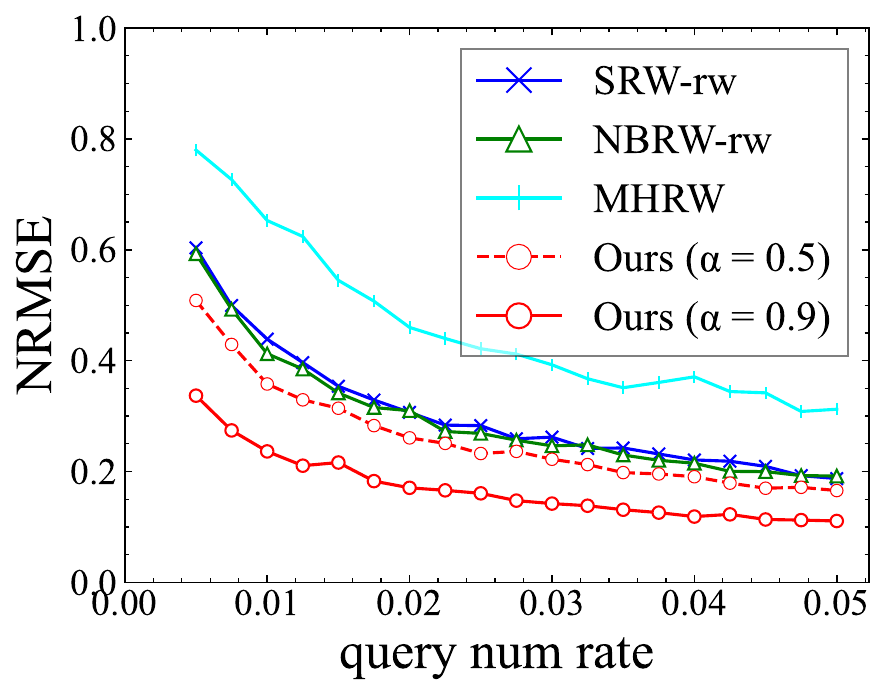}}
      \subfloat[Slashdot]{\includegraphics[scale=0.29]{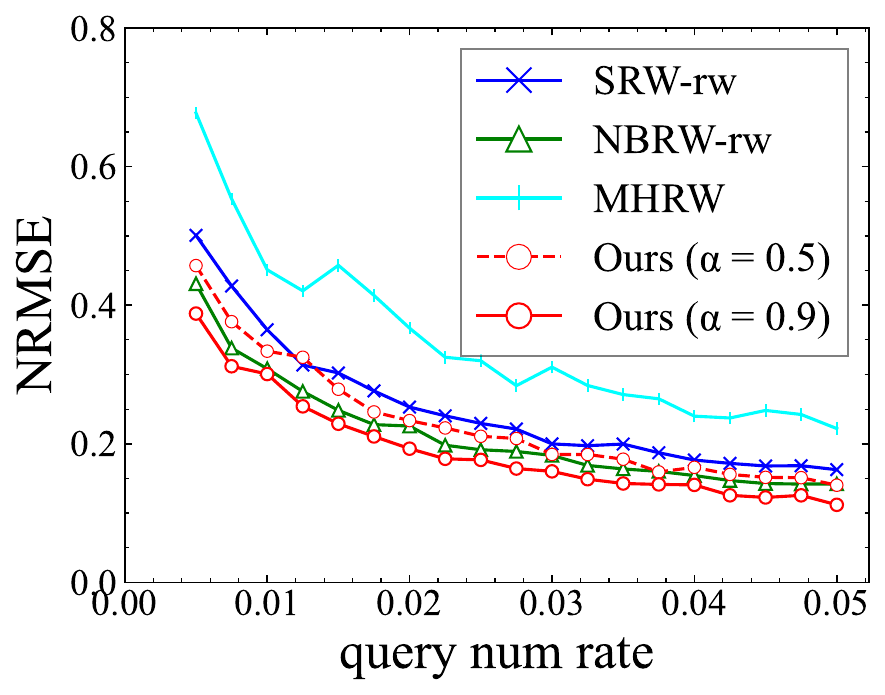}}
      \subfloat[Amazon]{\includegraphics[scale=0.29]{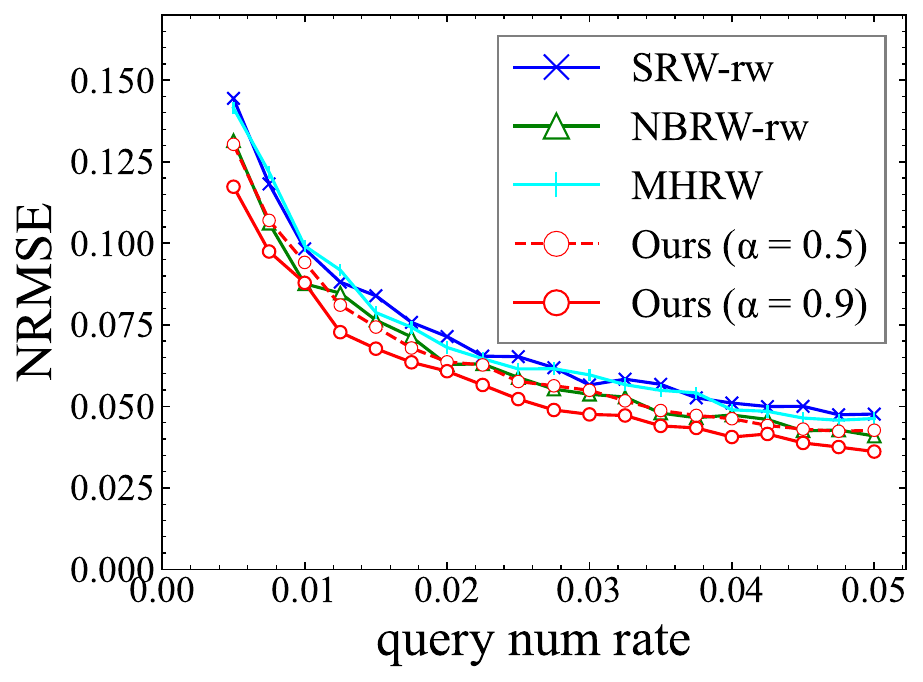}}
      \subfloat[DBA model]{ \includegraphics[scale=0.29]{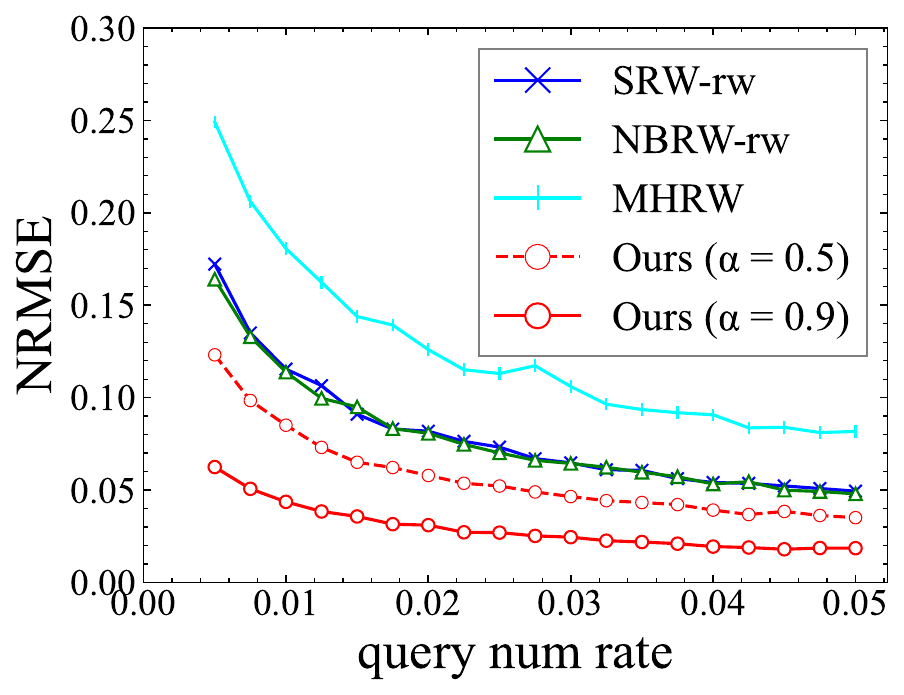}}
    \caption{NRMSE for low degree label estimation.}
    \label{fig:lowlabel} 
  \end{figure*}

\subsection{Comparison with Existing Methods}

We compared the estimation accuracy of each feature between the proposed and existing methods. Existing methods for comparison include well-known random walk-based feature estimation techniques: Simple Random Walk with Reweighting (SRW-rw)~\cite{gjoka2011practical, rasti2009respondent}, Non-backtracking Random Walk with Reweighting (NBRW-rw)~\cite{lee2012beyond}, and Metropolis-Hastings Random Walk (MHRW)~\cite{gjoka2011practical, rasti2009respondent}. These existing methods align with those compared by Iwasaki et al.~\cite{iwasaki2018comparing} under the same cost setting. Iwasaki et al.\ noted a tendency for NBRW-rw to achieve better accuracy when considering the acquisition frequency of adjacent nodes as the cost. However, they also highlighted the possibility of accuracy reversal based on the specific graph and features. 

SRW-rw weighs the sample sequence obtained from a regular random walk to estimate features. NBRW-rw estimates features by applying weighting to a sample sequence obtained from a random walk that avoids transitioning to the previously sampled node. MHRW adjusts transition probabilities based on node degrees, making the steady-state distribution uniform and ensuring that the expected values of the feature sequence from the sampled data serve as unbiased features. In this study, we compared the accuracy of feature estimation while maintaining a consistent query count, representing the number of times information from adjacent nodes is acquired.

Figures \ref{fig:degree}-\ref{fig:lowlabel} illustrate a comparison of the NRMSE values for each feature using the proposed method with $\alpha=0.5,0.9$ and existing methods. The horizontal axis represents the query count as a ratio to the total number of nodes. We vary the query count ratio from $0.25\%$ to $5\%$, adjusting in increments of 0.25\%. MHRW has been omitted owing to significant deviations, particularly in the average out-degree estimation of the DBA model. Across all graphs and features, the proposed method with $\alpha=0.9$ consistently matches or outperforms existing methods. As shown in Figure \ref{fig:alpha}, the proposed method achieves higher accuracy as $\alpha$ approaches $1$, but even at $\alpha=0.5$, it surpasses the existing methods for many graphs and features. Additionally, the results indicate that higher query counts $b$ lead to improved estimation accuracy.

\subsection{Discussion}
\begin{figure*}[t]
      \subfloat[ego-Twitter]{\includegraphics[scale=0.275]{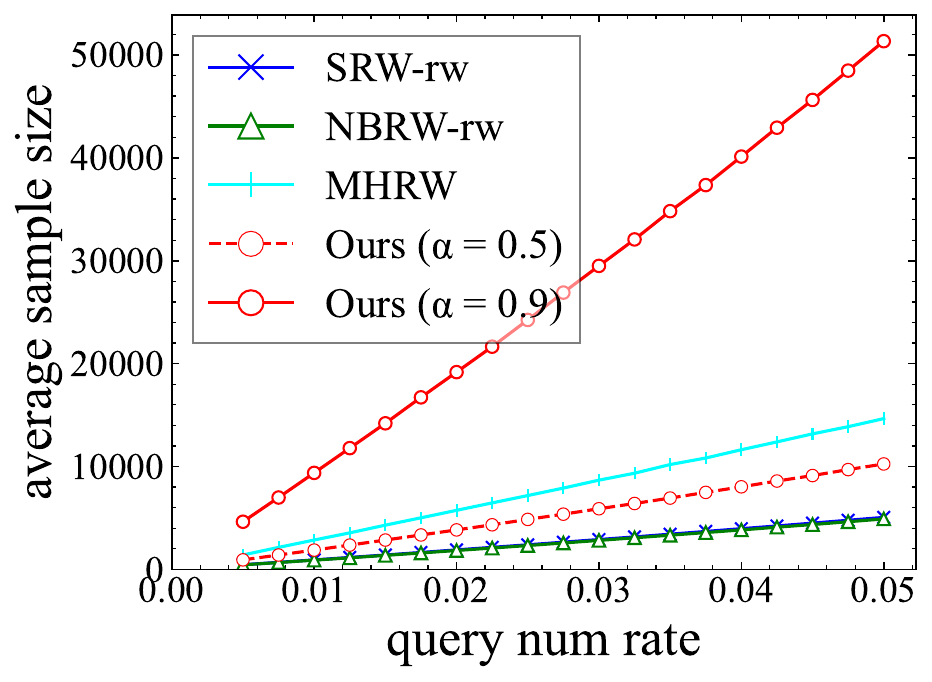}}
      \subfloat[Slashdot]{\includegraphics[scale=0.275]{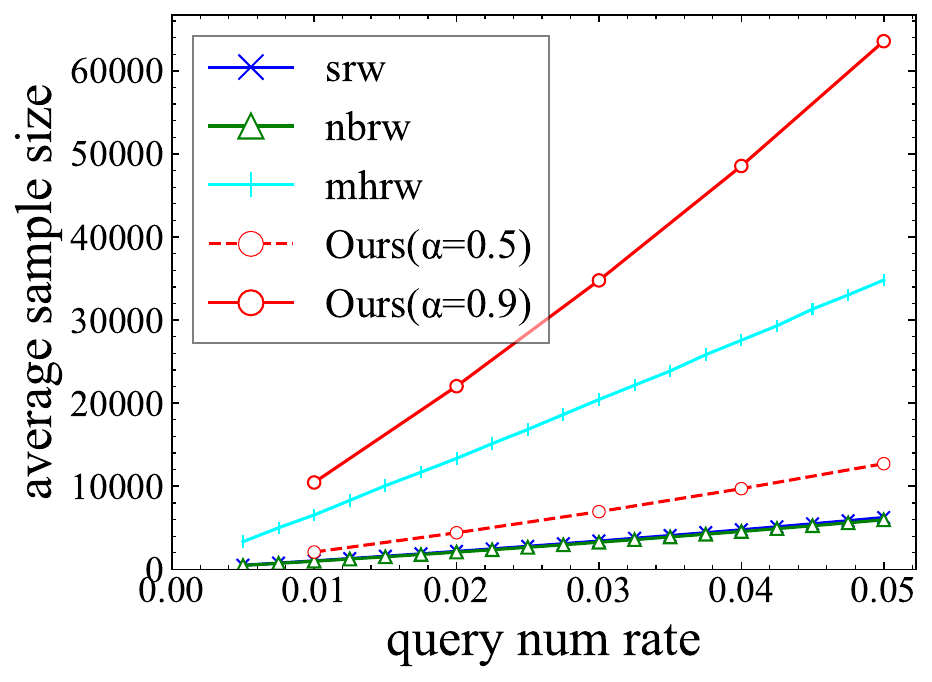}}
      \subfloat[Amazon]{\includegraphics[scale=0.275]{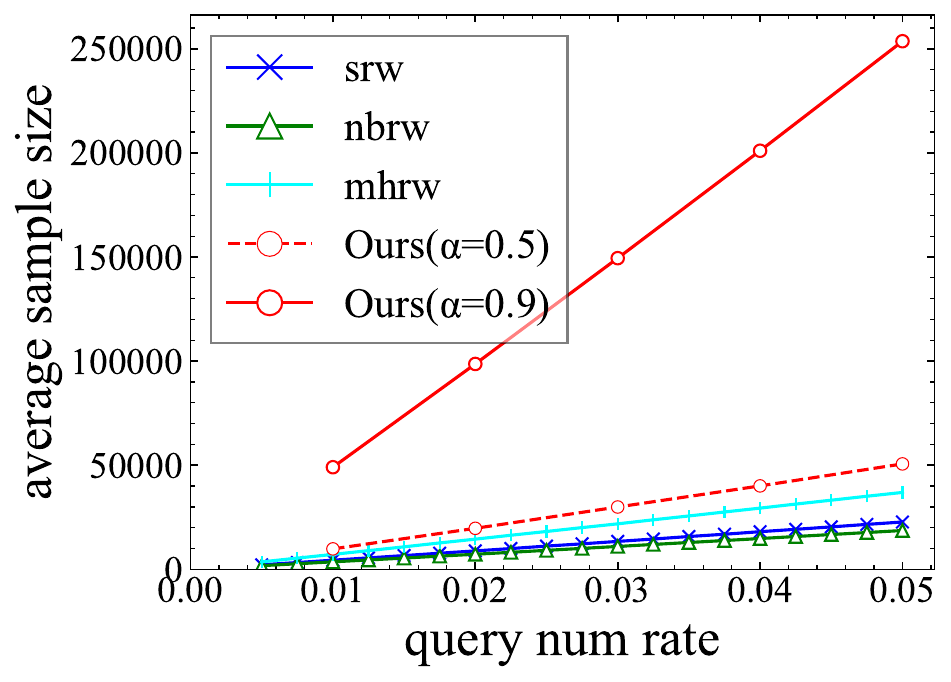}}
      \subfloat[DBA Model]{ \includegraphics[scale=0.275]{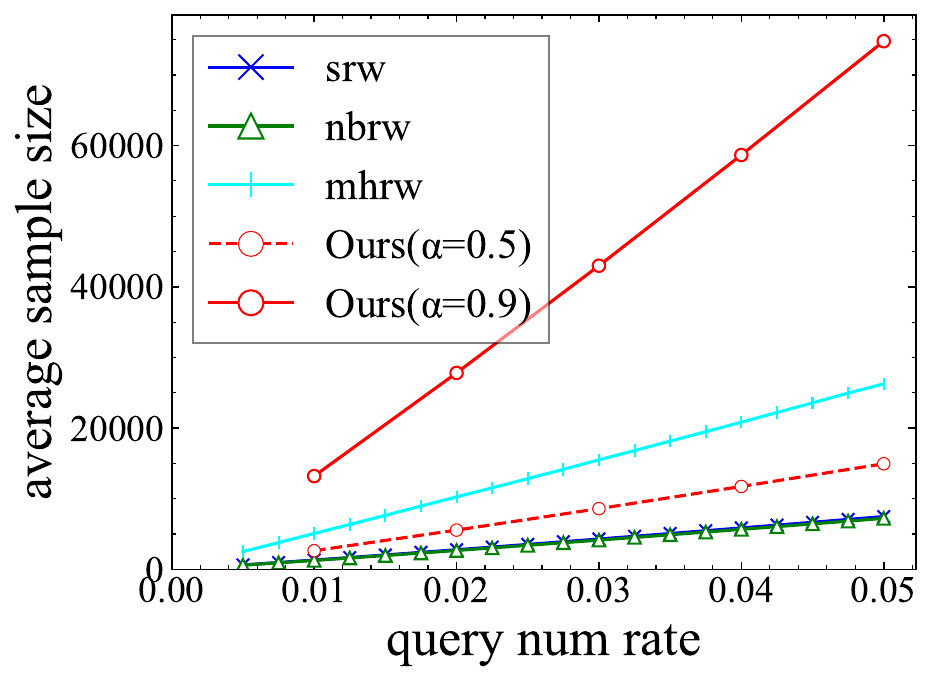}}
    \caption{Average sample sequence size for each query size.}
    \label{fig:samplesize} 
  \end{figure*}

In our proposed method, using more samples for estimation is expected to enhance the accuracy of the estimation. Considering the acquisition of neighboring nodes as a cost, the number of transition samplings within the same cost remains constant, regardless of $\alpha$. This quantity is equivalent to the number of transitions to new nodes in SRW. Our method allows an increase in the number of samples for adjacent node sampling available for estimation as $\alpha$ approaches 1. Therefore, as shown in Figure \ref{fig:alpha}, it is believed that estimation accuracy improves as $\alpha$ approaches 1. Figure \ref{fig:samplesize} shows the average sample size for the proposed method with parameters $\alpha=0.5, 0.9$, and existing methods concerning the number of queries. In the proposed method, setting $\alpha=0.9$ results in the highest number of samples available for estimation at the same cost compared to other methods. This is considered the reason for the superior accuracy of our proposed method, as illustrated in Figure \ref{fig:degree}-\ref{fig:lowlabel}.

The increase in sample size owing to adjacent node sampling is expected not to become a practical bottleneck. The estimator of our proposed method, $\sum_{s=1}^tw(Z_s)g(Z_s)/\sum_{s=1}^tw(Z_s)$ (Section \ref{subsec-est}), can be calculated sequentially during sampling, eliminating the need to store the entire sample sequence. The information that needs to be saved is, as mentioned in Section \ref{subsec-model}, the information of nodes once transited and the information of acquired adjacent nodes. Therefore, this stored information does not change based on the size of sample added by adjacent node sampling.

In practice, $\alpha$ should be determined by considering both the computational complexity and the API rate limit. The value of $\alpha$ affects the number of adjacent node samplings, which on average is $1/(1-\alpha)$ per API call (lines 8 of Algorithm \ref{alg:ours}). The sampling operation itself is efficient, taking $O(1)$ time (line 9-11 of Algorithm \ref{alg:ours}). Therefore, the average computational complexity of a adjacent node sampling for each API call is $O(1/(1-\alpha))$.

As $\alpha$ gets closer to 1, the complexity of sampling increases dramatically. Thus, $\alpha$ should be set based on the API's rate limit. For instance, if Mastodon's API allows one call per second~\cite{mastodonapi}, $\alpha$ should be chosen to be as large as possible while ensuring that the sampling process (line 6-18 of Algorithm \ref{alg:ours}) finishes within one second, depending on the system's computational power.

\section{Related work}
We discuss relevant research on graph sampling. Gjoka et al.~\cite{gjoka2011practical} compared RW-rw and MHRW and illustrated that SRW-rw achieved superior accuracy by reweighting from the steady-state distribution of random walks to obtain unbiased estimates. Lee et al.~\cite{lee2012beyond} introduced NBRW, a non-backtracking random walk that avoids revisiting the previous node. Through theoretical and experimental investigations, they established that NBRW-rw provided unbiased estimates and outperformed SRW-rw. Ribeiro et al.~\cite{ribeiro2010estimating} introduced the Multidimensional Random Walk, enhancing estimation accuracy in the presence of multiple connected components. While our study specifically focuses on graphs with a single weakly connected component, it is worth considering that applying the principles of the Multidimensional Random Walk may enhance accuracy in real-world applications. Iwasaki et al.~\cite{iwasaki2018comparing} proposed a method for comparing sampling algorithms based on query count and demonstrated the potential for accuracy evaluations of SRW-rw and NBRW-rw to reverse depending on the feature.

Next, we explore related research focusing on sampling techniques involving information from neighboring nodes. Han et al.~\cite{han2016waddling} introduced a method in which, during random walks, adjacent nodes are sampled based on the motif under estimation. The decision to acquire information from adjacent nodes depends on the motif, resulting in variations in the probability and depth of obtaining such information, distinguishing it from our study. Additionally, they consider the sampled node count as a cost, which differs from our study using API query count as the cost and estimating various features from adjacent nodes. Illenberger et al.~\cite{illenberger2012estimating} proposed a technique for estimating features by adjusting the sample sequence of snowball sampling and collecting information from adjacent nodes of accessed nodes. However, their approach assumes knowledge of the overall node count, distinguishing it from our study, which employs random walks to estimate features in unknown OSNs.

\section{conclusion}
In this study, we introduced a random walk that stochastically utilized information from adjacent nodes, considering the number of queries needed to obtain this information on OSNs as the cost. Through experiments, we demonstrated that our proposed method vielded more accurate estimations for average degree, the proportion of randomly assigned binary labels, the proportion of labels biased toward high-degree nodes, and the proportion of labels biased toward low-degree nodes compared to existing methods. Futhermore, experiments conducted across various datasets have revealed that accuracy improves as the parameter $\alpha$, representing the probability of adjacent node sampling, approaches 1.

\begin{acks}
This work was supported by JSPS KAKENHI Grant Number\linebreak JP21H04872.
\end{acks}

\bibliographystyle{ACM-Reference-Format}
\bibliography{hasegawa_graph}

\appendix
\section{Proof of Theorem \ref{the:sta}}\label{ape-sta}
First, we present the following lemma.
\begin{lemma}
    \label{lemma:mdsum}
    $\sum_{v_j\in N(v_i)}\frac{m(e_{ij})}{d_\mathrm{sum}(v_i)} = 1 $
\end{lemma}
\begin{proof}
\begin{align*}
    &\sum_{j\in N(v_i)}\frac{m(e_{ij})}{d_\mathbf{sum}(v_i)} \\
    &= \sum_{v_j\in N_\mathrm{out}(v_i)\backslash N_\mathrm{in}(v_i)}\frac{m(e_{ij})}{d_\mathrm{sum}(v_i)} + \sum_{v_j\in N_\mathrm{in}(v_i)\backslash N_\mathrm{out}(v_i)}\frac{m(e_{ij})}{d_\mathrm{sum}(v_i)} \\
    &\quad + \sum_{v_j\in N_\mathrm{out}(v_i) \cap N_\mathrm{in}(v_i)}\frac{m(e_{ij})}{d_\mathrm{sum}(v_i)} \\
    &=\frac{1}{{d_\mathrm{sum}(v_i)}}\cdot (1\cdot (d_\mathrm{out}(v_i) - d_\mathrm{in\mathchar`-out} (v_i)) \\
    &\quad+ 1\cdot (d_\mathrm{in}(v_i) - d_\mathrm{in\mathchar`-out}(v_i)) + 2\cdot (d_\mathrm{in\mathchar`-out}(v_i)))\\
    &= 1
\end{align*}
\end{proof}

In the following, we introduce Theorem \ref{the:sta}. To establish the necessity and sufficiency, it is enough to prove that $\boldsymbol{\pi} \mathbf{P} = \boldsymbol{\pi}$ based on Theorem \ref{the:mal}. Here, we define $\boldsymbol{\tilde{\pi}} \defeq \boldsymbol{\pi} \mathbf{P}$.

\noindent (i) The case $l \neq k$
\begin{align*}
    &\tilde{\pi}(e_{lk}) \\
    &= \sum_{v_j\in N(v_l)}\pi(e_{lj})\cdot p_{neighbor}(v_k;v_l)  \\
    &\quad+ \pi (e_{ll})\cdot p_{neighbor}(v_k;v_l) \\
    &= \sum_{v_j\in N(v_l)}\alpha \cdot \frac{m(e_{lj})}{2|E|}\cdot \alpha \cdot \frac{m(e_{lk})}{d_\mathrm{sum}(v_l)} \\
    & \quad+ (1-\alpha) \cdot \frac{d_\mathrm{sum}(v_l)}{2|E|} \cdot \alpha \cdot \frac{m(e_{lk})}{d_\mathrm{sum}(v_l)} \\
    &= \frac{\alpha ^2 \cdot m(e_{lk})}{2|E|}\cdot \sum_{v_j\in N(v_l)}\frac{m(e_{lj})}{d_\mathrm{sum}(v_l)} + \frac{\alpha(1-\alpha)\cdot m(e_{lk})}{2|E|} \\
    &= \frac{m(e_{lk})}{2|E|} (\alpha ^2 + \alpha(1-\alpha))\\
    &= \alpha \cdot \frac{m(e_{lk})}{2|E|} \\
    &= \pi (e_{lk})
\end{align*}
The first equality is valid owing to the probability transition matrix $\mathbf{P}$ defined in Definition \ref{def:p}. The second equality results from the application of Definition \ref{def:pn}. The third equality is derived through algebraic manipulation, and the fourth equality holds by applying Lemma \ref{lemma:mdsum}. The remaining equalities are also established through algebraic manipulation.

\noindent (i) The case $l = k$
\begin{align*}
    &\tilde{\pi}(e_{lk}) \\ 
    &= \sum_{v_i\in N(v_l)}\pi(e_{i,i})\cdot p_{walk}(v_l;v_i)  \\
    & \quad +\sum_{v_i\in N(v_l)}\sum_{v_j\in N(v_i)}\pi (e_{ij})\cdot p_{walk}(v_l;v_i)\\
    &= \sum_{v_i\in N(v_l)}(1-\alpha)\cdot \frac{d_\mathrm{sum}(v_i)}{2|E|}\cdot (1-\alpha)\cdot \frac{m(e_{il})}{d_\mathrm{sum}(v_i)} \\ 
    & \quad+ \sum_{v_i\in N(v_l)}\sum_{v_j\in N(v_i)}\alpha \cdot \frac{m(e{ij})}{2|E|}\cdot  (1-\alpha)\cdot \frac{m(e_{il})}{d_\mathrm{sum}(v_i)} \\
    &= \sum_{v_i\in N(v_l)} (1-\alpha)^2\cdot \frac{m(e_{il})}{2|E|} \\ 
    & \quad+ \sum_{v_i\in N(v_l)}\alpha(1-\alpha)\cdot \frac{m(e_{il})}{2|E|}\sum_{v_j\in N(v_i)}\frac{m(e_{ij})}{d_\mathrm{sum}(v_i)} \\
    &= \sum_{v_i\in N(v_l)}((1-\alpha)^2+\alpha(1-\alpha))\cdot \frac{m(e_{il})}{2|E|} \\
    &= (1-2\alpha + \alpha^2 + \alpha - \alpha^2)\cdot \frac{d_\mathrm{sum}(v_l)}{2|E|}\sum_{v_i\in N(v_l)}\frac{m(e_{li})}{d_\mathrm{sum}(v_l)} \\
    &= (1-\alpha)\cdot \frac{d_\mathrm{sum}(v_l)}{2|E|} \\
    &= \pi (e_{lk})
\end{align*}
The first equality is valid owing to the probability transition matrix $\mathbf{P}$ defined in Definition \ref{def:p}. The second equality results from the application of Definition \ref{def:pw}. The third equality is derived through algebraic manipulation, and the fourth equality holds by applying Lemma \ref{lemma:mdsum}. The fifth equality holds because $m(e_{ij}) = m(e_{ji})$, and the sixth equality follows from Lemma \ref{lemma:mdsum}.

As a result of conditions (i) and (ii), it is established that for any $e_{lk} \in \Omega$, the equality $\tilde{\pi}(e_{lk}) = \pi(e_{lk})$ holds true.

In conclusion, we have proven Theorem \ref{the:sta}.

\section{Proof of theorem \ref{the:fea}}\label{ape-fea}
To begin, consider the expected value of a function $g$ within the state space $\Omega$ concerning the stationary distribution $\boldsymbol{\pi}$, expressed as follows.

\begin{definition}
\label{def:epi}
    $\mathbb{E}_{\pi}(g) \stackrel{\mathrm{def}}{=} \sum_{e_{ij} \in \Omega}\pi(e_{ij})g(e_{ij})$
\end{definition}

Next, we define an estimator as follows.

\begin{definition}
\label{def:mu}
    $\hat{\mu _t}(g) \stackrel{\mathrm{def}}{=} \frac{1}{t}\sum_{s=1}^t g(Z_s)$
\end{definition}

Here, according to Theorem \ref{the:mar-str}, the following theorem holds.
\begin{theorem}
\label{the:lee}
    When the sequence ${Z_t}$ follows a finite and irreducible Markov chain with a stationary distribution $\boldsymbol{\pi}$, for any initial state, as $t$ approaches infinity,
    \[\hat{\mu _t}(g) \rightarrow \mathbb{E}_\pi (g) \text{ a.s.} \]
\end{theorem}

To prove Theorem \ref{the:fea}, we will introduce the following lemma.
\begin{lemma}
\label{lemma:wgtoeu}
    As $t$ approaches infinity,
    \[\frac{2|E|}{n}\cdot \hat{\mu_t}(wg) \rightarrow \mathbb{E}_u(f)\]
\end{lemma}

\begin{proof}
    \begin{align*}
        &\frac{2|E|}{n}\cdot \hat{\mu_t}(wg) \\ 
        &\rightarrow \frac{2|E|}{n}\cdot \mathbb{E}_{\pi}(wg)\\
        &=\frac{2|E|}{n}\cdot \sum_{e_{ij}\in \Omega}\pi(e_{ij})\cdot w(e_{ij}) \cdot g(e_{ij})\\
        &= \frac{2|E|}{n}\cdot\sum_{v_j\in V}\sum_{v_i\in N(v_j)}\pi(e_{ij})\cdot w(e_{ij}) \cdot f(v_j) \\ 
        &\quad+ \frac{2|E|}{n}\cdot\sum_{v_j\in V}\pi(e_{jj})\cdot w(e_{jj}) \cdot f(v_j) \\
        &= \sum_{v_j\in V}\sum_{v_i\in N(v_j)}\frac{2|E|}{n}\cdot \alpha \cdot \frac{m(e_{ij})}{2|E|}\cdot \frac{1}{d_\mathrm{sum}(v_j)}\cdot f(v_j) \\ 
        &\quad+ \sum_{v_j\in V}\frac{2|E|}{n}\cdot (1-\alpha) \cdot \frac{d_\mathrm{sum}(v_j)}{2|E|}\cdot \frac{1}{d_\mathrm{sum}(v_j)}\cdot f(v_j) \\
        &=\sum_{v_j\in V}\frac{\alpha}{n}\cdot f(v_j)\sum_{v_i\in N(v_j)}\frac{m(e_{ij})}{d_\mathrm{sum}(v_j)} + \sum_{v_j\in V}\frac{1-\alpha}{n}\cdot f(v_j)\\
        &=\frac{\alpha}{n} \sum_{v_j\in V} f(v_j) + \frac{1-\alpha}{n}\sum_{v_j\in V} f(v_j)\\
        &= \frac{1}{n}\sum_{v_j\in V}f(v_j)\\
        &=\mathbb{E}_u(f)
    \end{align*}
    The first equation is valid as per Theorem \ref{the:lee}. The subsequent equality is derived from the application of Definition \ref{def:epi}, while the second equation expands based on Definitions \ref{def:omg} and \ref{def:g}. The third equation involves the substitution of the stationary distribution $\boldsymbol{\pi}$, acquired from Theorem \ref{the:sta}, and the application of Definition \ref{def:w}. The fourth equation is a result of algebraic manipulation and the application of Lemma \ref{lemma:mdsum}. The remaining equations are established through algebraic transformations, and the final equation holds in accordance with Definition \ref{def:eu}.
\end{proof}

Furthermore, the following lemma is true.
\begin{lemma}
\label{lemma:wto1}
    As $t$ approaches infinity,
    \[\frac{2|E|}{n}\cdot \hat{\mu_t}(w) \rightarrow 1.\]
\end{lemma}
\begin{proof}
    Selecting a function $g$ with the property $g(e) = 1$ for any state $e \in \Omega,$ the equality holds through a transformation of expressions similar to that in Lemma \ref{lemma:wgtoeu}.
\end{proof}

Using Lemmas \ref{lemma:wgtoeu} and \ref{lemma:wto1}, as $t$ approaches infinity,
\begin{equation*}
    \frac{\sum_{s=1}^tw(Z_s)g(Z_s)}{\sum_{s=1}^tw(Z_s)} = \frac{\frac{2|E|}{n}\cdot \hat{\mu_t}(wg) }{\frac{2|E|}{n}\cdot \hat{\mu_t}(w) }\rightarrow \mathbf{E}_u(f) ~ \text{a.s.}
\end{equation*}

In conclusion, we have proven Theorem \ref{the:fea}.

\end{document}